\documentclass[journal,twoside,web]{ieeecolor}
\usepackage{generic}
\usepackage{cite}
\usepackage{amsmath,amssymb,amsfonts}
\usepackage{algorithmic}
\usepackage{graphicx}
\usepackage{textcomp}
\usepackage{subfig}
\usepackage{algorithm}
\newtheorem{theorem}{Theorem}

\newtheorem{lemma}{Lemma}
\newtheorem{remark}{Remark}
 
\newtheorem{assumption}{Assumption}

\def\BibTeX{{\rm B\kern-.05em{\sc i\kern-.025em b}\kern-.08em
		T\kern-.1667em\lower.7ex\hbox{E}\kern-.125emX}}
\begin{document}
	\title{Adaptive Tracking Control with Binary-Valued Output Observations}
	\author{Lantian Zhang and Lei Guo, \IEEEmembership{Fellow, IEEE}
		\thanks{This work was supported by the National Natural Science Foundation of China under Grant No. 12288201.}
		\thanks{Lantian Zhang and Lei Guo are with the Key Laboratory of Systems and Control, Academy of Mathematics and Systems Science, Chinese Academy of Sciences, Beijing 100190, China, and also with the School of Mathematical Science, University of Chinese Academy of Sciences, Beijing 100049, China. (e-mails: zhanglantian@amss.ac.cn, Lguo@amss.ac.cn). }}
	
	\maketitle
	
	\begin{abstract}
This paper considers real-time control and learning problems for finite-dimensional linear systems under binary-valued and randomly disturbed output observations. This has long been regarded as an open problem because the exact values of the traditional regression vectors used in the construction of adaptive algorithms are unavailable, as one only has binary-valued output information. To overcome this difficulty, we consider the adaptive estimation problem of the corresponding infinite-impulse-response (IIR) dynamical systems, and apply the double array martingale theory that has not been previously used in adaptive control. This enables us to establish global convergence results for both the adaptive prediction regret and the parameter estimation error, without resorting to such stringent data conditions as persistent excitation and bounded system signals that have been used in almost all existing related literature.  Based on this, an adaptive control law will be designed that can effectively combine adaptive learning and feedback control. Finally, we are able to show that the closed-loop adaptive control system is optimal in the sense that the long-run average tracking error is minimized almost surely for any given bounded reference signals. To the best of the authors' knowledge, this appears to be the first adaptive control result for general linear systems with general binary sensors and arbitrarily given bounded reference signals.
%
%
%
	\end{abstract}
	
	
		\section{Introduction} 

Binary output observation data, which refers to system output observations that can only take the values 0 or 1, is prevalent in practical systems.
One example comes from the signal detection problem in wireless communication \cite{rw2000,carli2010}, where signals are limited to 1-bit via a low-precision analog-to-digital converter, due to constraints in energy consumption and hardware complexity.  Another example is the binary classification problem in machine learning \cite{ba2014, mei2018}, where labels provide only discrete class information, rather than precise or continuous values. Furthermore, binary sensors have also been widely used in engineering applications, such as gas content sensors (CO2, H2, etc.) in the gas and oil industry \cite{sj2004}, and shift-by-wire switch sensors in automotive applications \cite{ls2001, ZG2003}.

Given their importance, extensive theoretical research has been devoted to estimation and control problems under such binary or quantized observations. 
Various parameter identification methods have been proposed for estimating unknown parameters using binary/quantized observation information, including the empirical measure approaches \cite{ZG2003}, maximum likelihood algorithms \cite{godoy}, set-membership methods \cite{MC2011}, stochastic approximation-type algorithms \cite{GZ2013, you2015}, and stochastic Newton algorithms \cite{ZZ2022, bercu}. Meanwhile, for the control problems under binary/quantized observations, numerous studies devoted to problems such as stabilization with quantized feedback observations (see, e.g.,\cite{nair2004, liberzon2014}) and consensus under quantized communications \cite{kash2007, yanlong2019}, under the assumption that the model parameters are known {\it a  prior}.
However, in many practical control scenarios, the models of systems are uncertain, one may require simultaneous parameter learning and real-time control. This basic adaptive control problem has rarely been explored in the existing literature, because it turns out that such a problem is highly non-trivial in both controller design and theoretical analysis due to the non-availability of the exact output observations. This is the primary motivation for us to initiate an investigation, by introducing new methods and establishing new results for the adaptive control of a basic class of finite-dimensional uncertain linear systems under binary output observations. 

For that purpose, one needs to build on some significant milestones of the existing adaptive control theory. In fact, over the past half a century, much progress has been made in the area of adaptive control. A natural approach in the design of adaptive control is the certainty equivalence principle, which involves two steps: (1) designing an online estimation algorithm using the observation data to estimate the unknown parameters, and (2) using these online estimated parameters to design the adaptive controller in place of the unknown true parameters. 
It is well-known that the nonlinear and complex nature of such closed-loop adaptive systems makes the corresponding theoretical analysis quite challenging and thus is regarded as a central issue in adaptive control. Take the basic adaptive control problem of linear stochastic systems as an example, motivated by the need to establish a rigorous theory for the least-squares(LS)-based self-turning regulators proposed by \AA str\"{o}m and Wittenmark \cite{str}, a great deal of effort had been devoted to the convergence study of LS with stochastic feedback signals resulting from stochastic adaptive control (\cite{ljung1976, moore1978, chris, lw1982, chen1982, goodwin, lai1986, kumar, g1991, g1995, cg1991}). Among the many significant contributions in this direction, we mention that Lai and Wei \cite{lw1982} established the asymptotic analysis of LS under the weakest possible excitation condition on the system stochastic signals, by using stochastic Lyapunov function methods and martingale convergence theorems. The convergence rates of either adaptive tracking control \cite{lai1986} or adaptive LQG control \cite{g1987} were established by resorting to certain kinds of external excitation signals for open-loop stable systems. The complete convergence theories of both the adaptive tracking control and adaptive LQG control were established later in Guo and Chen \cite{g1991, g1995}, Guo \cite{g1996}, and \cite{g1999}. These turn out to be useful foundations for the study of adaptive control under binary-valued observations in the current paper.

The aforementioned work all focuses on adaptive control problems under traditional continuous-valued output observations,  which cannot be directly applied to the current case of binary-valued output observations. In fact, even the adaptive parameter estimation theory under binary-valued output observations has not been fully established for general linear systems, let alone the investigation of the related adaptive control problems. This is primarily due to the following two reasons. Firstly, the regression vectors traditionally used for the design of adaptive learning algorithms are not available in the current scenario due to the availability of binary-valued observation information only. Consequently, almost all theoretical results on parameter estimation under binary output observations are limited to finite-impulse-response (FIR) systems with no poles in the linear models. Secondly, due to the nonlinearity of the current binary observation function, almost all existing adaptive estimation theories require stringent data conditions, such as independent and stationary conditions, deterministic persistent excitation (PE) conditions, and boundedness signal conditions, which are difficult or even impossible to verify for the closed-loop control system signals. For these reasons, only a few works have been devoted to adaptive control theory under quantized/binary-valued observations.  Guo et al. \cite{yanlong2011} considered the adaptive tracking problem for  FIR systems with a scalar unknown parameter. Zhao et al. \cite{yanlong2013} further investigated adaptive tracking of FIR systems under binary-valued observations for periodic tracking signals with full-rank conditions. In these two scenarios, it is possible to verify the PE and boundedness conditions required for parameter identification. However, such verification is challenging for infinite-impulse-response (IIR) systems and general reference signals.
Furthermore, Zhao et al.\cite{wenxiao2023} studied the adaptive regulation problem by using a stochastic approximation-based control algorithm, where the tracking signal is a fixed setpoint and the binary observation thresholds are specifically designed and depend on the given fixed setpoint. In summary, the adaptive control theory for general linear systems with both general binary sensors and time-varying reference signals remains a challenging research problem.

Partly motivated by the above open problem, we have established an adaptive estimation theory under binary-valued output observations and non-PE data conditions in our recent works \cite{ZZ2022, zztsqn}, but these results cannot be used to directly solve the abovementioned basic adaptive control problems, mainly due to the following two reasons: 1) the regression vector contains non-available exact output signals, and 2) the boundedness assumption of the input signals imposed in \cite{ZZ2022} is hard to verify for closed-loop control systems. To sidestep the difficulty of using the traditional regression vectors, we use the transfer function to transform the current system description into an IIR system description. Then, to establish the adaptive control theory of the current IIR system, we introduce the double array martingale theory developed by one of the authors in \cite{g1990} for the adaptive estimation theory of the $ARX(\infty)$ model, thereby establishing a convergence theory for parameter estimation without the need of PE and boundedness system signal conditions. Furthermore, in order to verify that the closed-loop system signals satisfy the required unbounded growth condition for adaptive estimation, we design an adaptive control by using a switching control technique depending on a certain level of signal excitation. All the above will make it possible for us to establish the stability and optimality of the closed-loop adaptive tracking control systems under binary-valued output observations. To the best of the authors' knowledge, this appears to be the first such result for adaptive control under binary output observations for general tracking signals.

The remainder of this paper is organized as follows. Section \ref{se2} formulates the problem and states the basic assumptions. Section \ref{se3} gives the main results of this paper, including the proposed estimation algorithms, the adaptive control laws, and the main theorems. Section \ref{se4} provides the proofs of the main results along with some key lemmas. Numerical examples are provided in Section \ref{se5}. Finally, we conclude the paper with some remarks.

\subsubsection*{Notations}  By $\|\cdot\|$, we denote the Euclidean norm of vectors or matrices. The maximum and minimum eigenvalues of a matrix $M$ are denoted by $\lambda_{max}\left\{M\right\}$ and $\lambda_{min}\left\{M\right\}$  respectively. Besides, by $det(M)$  we mean the determinant of the matrix $M$. Moreover, $\left\{\mathcal{F}_{k},k\geq 0\right\}$ is the sequence of $\sigma -$algebra together  with that of conditional mathematical expectation operator $\mathbb{E}[\cdot \mid \mathcal{F}_{k}]$.

\section{Problem formulation}\label{se2}	
Consider the following finite-dimensional linear systems with binary-valued output observations:
\begin{equation}\label{eq1}
\left\{
\begin{aligned}
		  &A(z)y_{n+1}=B(z)u_{n}\\
		&s_{n+1}=I\left(y_{n+1}+w_{n+1}>c_{n}\right)           
		\end{aligned} \right., \;\;\;\;\;n\geq 0,
\end{equation}
where $y_{n}, u_{n}, s_{n}$ and $w_{n}$ are the system output, input ($u_{k}=0, \forall k<0$), output observation and noise, respectively. $A(z)$ and $B(z)$ are polynomials in the backward-shift operator $z$:
\begin{equation}
\begin{aligned}
A(z)=&1+a_{1}z+\cdots+a_{p}z^{p},\;\;p\geq 0,\\
B(z)=&b_{1}+b_{2}z+\cdots+b_{q}z^{q-1},\;\;q\geq 1.
\end{aligned}
\end{equation}
Unlike in traditional precise observation scenarios, the system output cannot be obtained; instead, only two values can be observed. $I\left(\cdot\right)$ is the indicator function, and $\{c_{n}\} $ denotes a given threshold sequence.

The problem considered here is, at any time instant $n\geq 0$, to design a feedback control $u_{n}$ based only on the past binary-valued output observations $\{s_{0}, s_{1}, \cdots, s_{n}\}$ and the known reference signal $y_{n+1}^{*}$, such that the following averaged tracking error between the unknown system output $y_{n+1}$ and the reference signal $y_{n+1}^{*}$ can asymptotically achieve its minimum:
\begin{equation}
J_{n}=\frac{1}{n}\sum_{i=1}^{n}\left(y_{i}-y_{i}^{*}\right)^{2}.
\end{equation}

To proceed with further discussions, we need the following assumptions.

\begin{assumption}\label{ass3} Both $A(z)$ and $B(z)$ are stable polynomials, that is $A(z)\not=0, B(z)\not=0, \forall |z|\leq 1.$
\end{assumption}

\begin{assumption}\label{ass4}
The reference signal $\{y_{k+1}^{*}, \mathcal{F}_{k}\}$ is a known bounded adapted sequence (where $\{\mathcal{F}_{k}\}$ is a sequence of nondecreasing $\sigma-$algbras).
\end{assumption}

\begin{assumption}\label{ass2}
The threshold $\{c_{k},\mathcal{F}_{k}\}$ is a known bounded adapted sequence, with $\|c_{k}\|\leq C$, where $C$ is a positive constant.
\end{assumption}

\begin{assumption}\label{ass1}
The noise $w_{k}$ is $\mathcal{F}_{k}-$measurable. For any $k\geq 1$, the noise $w_{k}$ is independent with $\mathcal{F}_{k-1}$, and its probability density function, denoted by $f_{k}(\cdot)$, is known and satisfies
		\begin{equation}\label{66}
			\begin{aligned}
				\inf_{k\geq 1, x\in \mathfrak{X}} \left\{f_{k}(x)\right\}> 0,
			\end{aligned}
		\end{equation}
for any bounded set $\mathfrak{X}\subset \mathbb{R}.$		\end{assumption}

\begin{remark}
It can be easily seen that if the noise sequence $\{w_{k}, k\geq 1\}$ follows an i.i.d Gaussian distribution as assumed previously (see, e.g.,\cite{godoy, GZ2013}), then (\ref{66}) holds. Assumption \ref{ass1} has been widely used in previous theoretical studies on adaptive identification under binary-valued output observations with given thresholds (see, e.g., \cite{ZG2003, ZZ2022, zztsqn}). 
\end{remark}

It is obvious that the minimum of $\lim\limits_{n\to \infty} J_{n}$ is $0$, and the corresponding optimal control law satisfies 
\begin{equation}\label{e5555}
y_{k+1}^{*}=\mathbb{E}\left[y_{k+1}\mid \mathcal{F}_{k}\right]=\varphi_{n}^{\top}\vartheta,
\end{equation}
 where $\vartheta$ is the unknown parameter defined by
\begin{equation}
\vartheta=\left[-a_{1}, \cdots, -a_{p}, b_{1},\cdots,b_{q}\right]^{\top},
\end{equation}
and $\varphi_{n}$ is the corresponding regressor
\begin{equation}\label{e6}
\varphi_{n}=\left[y_{n}, \cdots, y_{n-p+1}, u_{n}, \cdots, u_{n-q+1}\right]^{\top}.
\end{equation}
In traditional scenarios where the parameter $\vartheta$ is unknown but the outputs $\{y_{k}\}$  are available, a natural approach is first to provide an estimate $\hat{\vartheta}_{k}$ for the unknown parameter $\vartheta$ at each time instant $k$ based on the past measurements $\{y_{i}, \varphi_{i-1}, 1\leq i \leq k\}$, and then design the adaptive control by solving the equation 
\begin{equation}\label{ex7}
y_{k+1}^{*}=\varphi_{k}^{\top}\hat{\vartheta}_{k}. 
\end{equation}
However, in the current scenario where the system outputs $\{y_{k}\}$ are unavailable, since the regression vector $\{\varphi_{k}\}$ cannot be obtained, it turns out to be impossible to estimate the unknown parameter $\vartheta$ using $\{\varphi_{k}\}$ directly and to design the adaptive control law by solving equation $(\ref{ex7})$ directly. 
To overcome this difficulty, we consider the transfer function:
\begin{equation}\label{eE8}
G(z)=A^{-1}(z)B(z)=\sum_{i=1}^{\infty}G_{i}z^{i-1},
\end{equation}
where 
\begin{equation}\label{pro}
G_{1}=b_{1}, \;\; \|G_{i}\|=O\left(e^{-ri}\right), r>0, i>1.
\end{equation}
The system can then be transformed into an IIR system as follows:
\begin{equation}\label{e88}
y_{n+1}=\sum_{i=1}^{\infty}G_{i}u_{n-i+1}.
\end{equation}
For the system described by $(\ref{e88})$, we introduce the parameter vector
\begin{equation}\label{e11}
\theta=[G_{1}, G_{2},\cdots]^{\top}\in \mathbb{R}^{\infty},
\end{equation}
and the corresponding regression vector
\begin{equation}
\phi_{n}=\left[u_{n}, u_{n-1},\cdots,\right]^{\top}\in \mathbb{R}^{\infty},
\end{equation}
so that $(\ref{eq1})$ can be rewritten as
\begin{equation}\label{e13}
\left\{
\begin{aligned}
		  &y_{n+1}=\phi_{n}^{\top}\theta\\
		&s_{n+1}=I\left(y_{n+1}+w_{n+1}>c_{n}\right)           
		\end{aligned} \right., \;\;\;\;\;n\geq 0.
\end{equation}
According to $(\ref{e13})$, the equation $(\ref{e5555})$ that the optimal control satisfies can be rewritten in the following form:
\begin{equation}\label{opp}
y_{n+1}^{*}=\mathbb{E}\left[y_{n+1}\mid \mathcal{F}_{n}\right]=\phi_{n}^{\top}\theta.
\end{equation}
Based on $(\ref{opp})$, the optimal control sequence $\{u_{n}^{0}\}$ can then be explicitly expressed as
\begin{equation}\label{e14}
u_{n}^{0}=\frac{1}{G_{1}}\left(y_{n+1}^{*}-\sum_{i=2}^{\infty}G_{i}u_{n-i+1}^{0}\right).
\end{equation}
Since the parameters $\{G_{i}\}$ are unknown, a natural problem is how to design an estimation algorithm for this infinite-dimensional parameter vector by using binary-valued output observations only, which is needed in the design of an adaptive feedback control law. This issue will be discussed in the subsequent sections.

\section{Main results}\label{se3}	

In this section, we first introduce an adaptive estimation algorithm for the unknown parameter vector $\theta$. Compared to existing algorithms, a notable characteristic of the new algorithm is that we can establish its asymptotic properties without requiring any boundedness conditions and non-PE excitation conditions of the regression vectors. Furthermore, based on this estimation algorithm, we design a switching adaptive control and establish its global stability and asymptotic optimality.

\subsection{Adaptive parameter estimation}

To estimate the parameter vector in the current ``large model", we introduce sequences of vectors with progressively increasing dimensions.
Define the dimension sequence $\{p_{n}, n\geq 1\}$ as follows:
\begin{equation}\label{pn}
p_{n}=\lfloor\log^{a}n\rfloor,\;\; a>1, 
\end{equation}
where $\lfloor x\rfloor$ denotes the integer part of $x$. Correspondingly, we define a sequence of parameter vectors $\{\theta(n), n\geq 1\}$ with dimensions growing at a rate determined by $p_{n}$:
\begin{equation}\label{e17}
\theta(n)=\left[G_{1}, G_{2}, \cdots, G_{p_{n}}\right]^{\top},
\end{equation}
where $\{G_{i}, i\geq 1\}$ are defined in $(\ref{e11})$, and also define the sequence of regression vectors $\{\phi_{k}(n), 1\leq k\leq n\}$ as:
\begin{equation}
\phi_{k}(n)=\left[u_{k}, u_{k-1}, \cdots, u_{k-p_{n}+1}\right]^{\top}.
\end{equation}

Moreover, to constrain the growth rate of parameter estimation, we introduce a projection operator into the algorithm. For this, we first introduce a function $g(t): [0, \infty)\rightarrow \mathbb{R}$  defined as follows: 
\begin{equation}\label{e7}
\begin{aligned}
g(t)=\inf\limits_{k\geq 1}\inf\limits_{\|x\|\leq t+C}f_{k+1}(x),\;\; t\geq 0.
 \end{aligned}
 \end{equation}
Then it follows that the function $g(t)$ is non-increasing and satisfies $g(t)>0, \forall t\geq 0$ by Assumption \ref{ass1}. 
Consequently, we can define the inverse function $h(\cdot)$ of $g(\cdot)$ as follows:
\begin{equation}\label{e8}
\begin{aligned}
h(t)=\inf\left\{s\geq 0: g(s)=t\right\}, \;\;0<t\leq g(0),
\end{aligned}
\end{equation}
so that $h(t)$ is also a non-increasing function. 

Now, for each $1\leq k \leq n$, define the projection region as follows:
\begin{equation}
D_{k}(n)=\left\{x\in \mathbb{R}^{p_{n}}: \|x\|_{1}\leq d_{k}\right\},
\end{equation}
where $\| \cdot\|_{1}$ is the $L_{1}$ norm, the upper bound $d_{k}$ is defined by 
\begin{equation}\label{e21}
d_{k}=\sqrt{h\left(\frac{g(0)}{\alpha_k}\right)+1},
\end{equation}
and $\{\alpha_{k}, k\geq 1\}$ is any positive increasing sequence satisfying
\begin{equation}\label{22}
\alpha_{k}\to \infty,\;\;\;\alpha_{k}^{2}=o\left(\sqrt{k}\right),\;\;a.s.,\;\; k\to \infty.
\end{equation}
The adaptive parameter estimates are given by Algorithm \ref{alg1}.
\begin{algorithm}
\caption{} 
\label{alg1}
For each $n\geq 1$, the parameter estimates are recursively defined for $0\leq k\leq {n-1}$ as follows: 
\begin{equation}\label{9}
\begin{aligned}
			\theta_{k+1}(n)&=\Pi_{P_{k+1}^{-1}(n)}^{D_{k}(n)}\left\{\theta_{k}(n)+a_{k}(n)\beta_{k}P_{k}(n)\phi_{k}(n)e_{k+1}(n)\right\},\\
						e_{k+1}(n)&=s_{k+1}-1+F_{k+1}(c_{k}-\phi_{k}^{\top}(n)\theta_{k}(n)),\\
			P_{k+1}(n)&=P_{k}(n)-\beta_{k}^{2}a_{k}(n)P_{k}(n)\phi_{k}(n)\phi_{k}^{\top}(n)P_{k}(n),\\
			a_{k}(n)&=\frac{1}{1+\beta_{k}^{2}\phi_{k}^{\top}(n)P_{k}(n)\phi_{k}(n)},\\
			\beta_{k}&=\inf_{|x|\leq d_{k}\max\limits_{1\leq i \leq k}\left|u_{i}\right|+C}\{f_{k+1}(x)\},
			\end{aligned}
		\end{equation}
		where $F_{k+1}(\cdot)$ is the probability distribution function of $w_{k+1}$; $d_{k}$ and C are defined as in $(\ref{e21})$ and Assumption \ref{ass2}; $\Pi_{P_{k+1}^{-1}}^{D_{k}(n)}\{\cdot\}$ is the projection operator defined via a time-varying Frobenius norm as follows:
	\begin{equation}	
		 \Pi_{P_{k+1}^{-1}(n)}^{D_{k}(n)}\{x\}=\mathop{\arg\min}_{\omega\in D_{k}(n)}(x-\omega)^{\top}P_{k+1}^{-1}(n)(x-\omega), \forall x \in \mathbb{R}^{p_{n}},
		 \end{equation}	
		 the initial values $\theta_{0}(n)$ and $P_{0}(n)>0$ can be chosen arbitrarily.	 
\end{algorithm}

\begin{remark}
Algorithm \ref{alg1} is inspired by the estimation algorithms for the $ARX(\infty)$ model in \cite{g1990} and \cite{g2023}. Its recursive form is similar to the recursive least squares algorithm for linear stochastic models, with the main difference being the design of the update term $e_{k+1}$, which utilizes only binary observation information $s_{k+1}$.
Compared to the estimation algorithm with binary-valued observations in \cite{g2023}, the current algorithm differs in two key aspects: Firstly, it does not require prior knowledge of the bounded convex set containing the parameter $\theta$, which is not easily obtainable for the current model. Secondly, in designing the scalar gain $\beta_{k}$,  it does not rely on the prior assumption of an upper bound of the system input, which is crucial for establishing adaptive control theory later on.
\end{remark}

Based on the adaptive estimates of the unknown parameter, for each $k\geq 1$, we can define the following adaptive predictor for the output:
\begin{equation}
\hat{y}_{k+1}=\phi_{k}(k)^{\top}\theta_{k}(k),
\end{equation}
where $\theta_{k}(k)$ is given by Algorithm \ref{alg1}. Here, the difference between the real output $y_{k+1}$ and the adaptive predictor $\hat{y}_{k+1}$ is referred to as regret $R_{k}$, which is defined as follows:
\begin{equation}\label{rr}
R_{k}=\left[y_{k+1}-\phi_{k}(k)^{\top}\theta_{k}(k)\right]^{2}.
\end{equation}
The asymptotic properties of the regrets are crucial for the theoretical analysis of adaptive control, which will be established in the following theorem.

\begin{theorem}\label{le1}
Under Assumptions $\ref{ass3}$-$\ref{ass1}$, suppose for some $\delta>0$,
\begin{equation}\label{ee25}
u_{n}=O(\min\{d_{n}, n^{1+\delta}\}), \;\;a.s.,\;\; n\to \infty,
\end{equation}
where $d_{n}$ is defined in $(\ref{e21})$. Then the following properties hold:
\begin{itemize}
\item The sample paths of the accumulated regrets have the following asymptotic upper bound as $n\to \infty$:
\begin{equation}\label{e27}
\begin{aligned}
\sum_{k=0}^{n}R_{k}
=O\left(\alpha_{n}^{2}p_{n}^{2}\log n\right),\;\;a.s.,
\end{aligned}
\end{equation}
where $p_{n}$, $\alpha_{n}$, and $R_{n}$ are defined as in $(\ref{pn})$, $(\ref{22})$, and $(\ref{rr})$, respectively.  
\item The parameter estimates given by Algorithm \ref{alg1} have the following asymptotic upper bound as $n\to \infty$:
\begin{equation}\label{6}
\left\|\theta(n)-\theta_{n}(n)\right\|^{2}=O\left(\frac{\alpha_{n}^{2}p_{n}\log n}{\lambda_{\min}(n)}\right), a.s.,
\end{equation}
where $\lambda_{\min}(n)$ is defined by
\begin{equation}\label{lam}
\lambda_{\min}(n)=\lambda_{\min}\left\{P_{0}^{-1}+\sum\limits_{k=0}^{n}\phi_{k}(n)\phi_{k}^{\top}(n)\right\}.
\end{equation}
\end{itemize}
\end{theorem}

The proof is supplied in the next section.

The results of the accumulated regrets in $(\ref{e27})$ do not require any excitation condition, but a certain growth rate condition $(\ref{ee25})$ on the system signals is required, which will be verified for the closed-loop adaptive systems in the proof of Theorem \ref{thm2} in the next section. 

\begin{remark}
If $A(z)$ and $B(z)$ have no common factor and $\left[a_{p}, b_{q}\right]\not=0$, it is well known that the linear equations connecting $\{a_{i}, b_{j}, 1\leq i\leq p, 1\leq j \leq q\}$ with $\{G_{i}, i\geq 1\}$ are as follows:
\begin{equation}
b_{i}=\sum_{j=0}^{i \land p}a_{j}G_{i-j},\;\;\forall 1\leq i\leq q,
\end{equation}
where $\land$ represents the minimum value, and
\begin{equation}\label{var}
\left[a_{1},\cdots,a_{p}\right]^{\top}=-\left(LL^{\top}\right)^{-1}L\left[G_{q+1},\cdots, G_{q+p}\right]^{\top},
\end{equation} 
where 
\begin{equation}
L=\left[\begin{matrix} 
G_{q} & G_{q+1} &\cdots & G_{q+p-1}\\ 
G_{q-1} & G_{q} &  \cdots  & G_{q+p-2}\\
\cdots & \cdots &\cdots & \cdots \\
G_{q-p+1} & G_{q-p+2} &\cdots &G_{q} 
\end{matrix}\right],
\end{equation}
and $G_{i}\overset{\triangle}{=}0$ for $i<0$. Replacing $\left\{G_{i}, i\geq 1\right\}$ in $(\ref{var})$ with their estimates $\{\theta_{n}(n),  n \geq 1\}$ generated by Algorithm \ref{alg1}, we can obtain the estimates $\{\vartheta_{n}, n\geq 1\}$ for the unknown parameter $\vartheta$. From $(\ref{6})$, if we have
\begin{equation}\label{ee22}
\alpha_{n}^{2}p_{n}\log n=o\left(\lambda_{\min}(n)\right), a.s.,
\end{equation}
as $n\to \infty$, then the estimates $\vartheta_{n}$ will be strongly consistent, i.e., $\lim\limits_{n\to \infty}\|\vartheta-\vartheta_{n}\|=0,\;a.s.$ The condition $(\ref{ee22})$ does not need the input signals to be PE and will be used in the design of adaptive control in the next section.
\end{remark}

\subsection{Adaptive tracking control}

\subsubsection{Main ideas}
Although we have relaxed the boundedness conditions on the system inputs in the asymptotic analysis of the parameter estimation algorithms, Theorem \ref{le1} still imposes a growth rate constraint on the system inputs. Unfortunately, guaranteeing such a constraint theoretically appears to be quite hard for a control law designed directly using the certainty equivalence principle, i.e.,
\begin{equation}\label{e16}
u_{n}^{0}=\frac{1}{G_{1n}}\left(y_{n+1}^{*}-\sum_{i=2}^{\infty}G_{in}u_{n-i+1}^{0}\right),
\end{equation}
which prompts us to make some modifications to the controller design.

To motivate the ideas for improving the controller design, we first considered the case where the parameter $\theta$ is known. Since $G(z)$ is stable, it is easy to verify that the input sequence $\{u_{n}, n\geq 1\}$ generated by $(\ref{e14})$ is a bounded sequence. Furthermore, if the parameter estimation $\theta_{n}(n)$ is located within a certain ``stable neighborhood" of the true parameter $
\theta$, it can also be proven that the input sequence keeps bounded. See Lemma \ref{lem22} below for details.

Based on the above analysis, for the case where the parameter is unknown, if the input value $u_{n}$ given by $(\ref{e16})$ grows too fast, it indicates that we may not have sufficient signal information to drive the parameter estimates into a certain neighborhood of the true parameter, and in this case, online parameter learning may be called for. On the other hand, according to Theorem \ref{le1}, if the minimum eigenvalue of the information matrix, i.e., $\lambda_{\min}(n)$, is sufficiently large, then there is adequate signal information to ensure the parameter estimates enter into a ``stable neighborhood", leading to the stability of the adaptive control. These two levels of signal information for estimation inspire a switching adaptive control between adaptive learning and feedback in the initial phase, but it will be proven that after some finite time, the adaptive controller will rarely be switched to the learning phase. 

We now proceed to introduce the specific design of adaptive control.


\subsubsection{Adaptive controller}
Firstly, to deal with the potential occurrence of $b_{1n}=0$ in $(\ref{e16})$, we make the necessary modification to the estimates produced by Algorithm \ref{alg1}. 
Let
\begin{equation}\label{166}
\begin{aligned}
\hat{\theta}_{n}(n)&=\theta_{n}(n)+\left(a_{n}(n)P_{n}(n)\right)^{\frac{1}{2}}e_{i_{n}}\\
&=\left[\hat{G}_{1n}, \hat{G}_{2n}, \cdots,\hat{G}_{p_{n}n}\right],
\end{aligned}
\end{equation}
where $\{i_{n}\}$ is a sequence of integers defined by
\begin{equation}\label{177}
i_{n}=\mathop{\arg\max}_{0\leq i\leq p_{n}}\left|G_{1n}+e_{1}^{\top}P_{n}^{\frac{1}{2}}(n)e_{i}(n)\right|,
\end{equation}
with $e_{0}(n)=0$, and $e_{i}(n), 1\leq i\leq p_{n}$ representing the $i$th column of the $p_{n}\times p_{n}$ identity matrix. Consequently, the modified adaptive control law is given by
\begin{equation}\label{14}
u_{n}^{s}=\frac{1}{\hat{G}_{1n}}\left(y_{n+1}^{*}-\hat{G}_{2n}u_{n-1}\cdots-\hat{G}_{p_{n}n}u_{n-p_{n}+1}\right).
\end{equation}
To make the adaptive system have some guaranteed level of excitation, we now introduce an arbitrary i.i.d. sequence $\{v_{t}\}$ independent of $\{w_{t}\}$ with properties
\begin{equation}\label{e388}
\mathbb{E}[v_{k}]=0, \;\;\mathbb{E}[v_{k}^{2}]=1,\;\; \|v_{k}\|\leq \bar{v},
\end{equation}
where $\bar{v}$ is a constant. 
Let us partition the time axis by a sequence of stopping times
\begin{equation}
1=\tau_{1}<\sigma_{1}<\tau_{2}<\sigma_{2}<\cdots
\end{equation}
at which the adaptive controller $u_{t}$ will be switched, i.e.,
\begin{equation}\label{control}
u_{t}=\left\{
		\begin{array}{rcl}
		 			&u_{t}^{s},\;\;\;           & \text{if}\; t\; \text{belongs to some} [\tau_{k}, \sigma_{k}) \\       
					 &v_{t},\;\;\; &\text{if}\; t\; \text{belongs to some} [\sigma_{k}, \tau_{k})  
		\end{array} \right.,
\end{equation}
where
\begin{equation}\label{201}
\begin{aligned}
\sigma_{k}=
\sup\left\{t>\tau_{k}: |u_{j}^{s}|\leq \min\{d_{j}, j^{b}\}+\left|u_{0}^{s}\right|,\;\forall j \in [\tau_{k},t)\right\}
\end{aligned}
\end{equation}
\begin{equation}\label{200}
\begin{aligned}
&\tau_{k+1}=
\inf\left\{t>\sigma_{k}: \lambda_{\min}(t)\geq \alpha_{t}^{2}p_{t}\log^{1+\epsilon}t \vphantom{\max_{t-p_{t}\leq j \leq t-1}\{|u_{j}|, |u_{t}^{s}|\}\leq \min\{\sqrt{d_{t}}, t^{\frac{b}{2}}\} }\right., \\
&\quad\quad\;\;\;\;\;\quad\left.\max_{t-p_{t}\leq j \leq t-1}\left\{|u_{j}|, |u_{t}^{s}|\right\}\leq \min\{\sqrt{d_{t}}, t^{\frac{b}{2}} \} \right\},\\
\end{aligned}
\end{equation}
where $b$ and $\epsilon$ are fixed constants with $b\in(0,\frac{1}{4})$ and $\epsilon>0$.

We now make a brief explanation about the switching adaptive control $(\ref{control})$. By $(\ref{control})$-$(\ref{200})$, it is easy to see that from the random time $\tau_{k}$, the adaptive control $u_{t}$ is taken as $u_{t}^{s}$ as far as $t<\sigma_{k}$, where $\sigma_{k}$ is the first time when the growth rate of $|u_{t}^{s}|$ is greater than $\min\{d_{t}, t^{b}\}+|u_{0}^{s}|$; then from the random time $\sigma_{k}$, the adaptive control is taken as $v_{t}$, until both the growth rate of the the minimum eigenvalue $\lambda_{\min}(t)$ is larger than $\alpha_{t}^{2}p_{t}\log^{1+\epsilon} t$ and the growth rate of $|u_{t}^{s}|$ is less than $\min\{\sqrt{d_{t}}, t^{\frac{b}{2}}\}$.
Moreover, according to $(\ref{control})$-$(\ref{200})$, it can easily been seen that for any $t\geq 0$, the growth rate of the adaptive control $u_{t}$ is constrained within the following bound:
\begin{equation}
|u_{t}|\leq \min\left\{d_{t}, t^{b}\right\}, t\geq 0.
\end{equation}

\begin{theorem}\label{thm2}
Let Assumptions $\ref{ass3}$-$\ref{ass1}$ be satisfied. Then the control system $(\ref{eq1})$ under adaptive control law $(\ref{control})$-$(\ref{200})$ is globally asymptotically optimal, and has the following convergence rate (regret of adaptive tracking) as $n\rightarrow \infty$:
\begin{equation}
\sum_{k=1}^{n}(y_{k}^{*}-y_{k})^{2}=O\left(n^{\delta}p_{n}\right),\;\;\;\;a.s.
\end{equation}
for any $ \delta \in \left(\frac{1}{2}+2b, 1\right)$, where $p_{n}$ and $b$ are defined in $(\ref{pn})$ and $(\ref{201})$, $(\ref{200})$, respectively.
\end{theorem}

	\section{Proof of main results}\label{se4}	
	\subsection{Proof of Theorem \ref{le1}}
	We first introduce the following lemma on double array martingale theory:
	\begin{lemma}\label{gannals}(\cite{gannals})
Let $\{w_{t}, \mathcal{F}_{t}\}$ be a martingale difference sequence satisfying
\begin{equation}
\sup_{t}\mathbb{E}\left[\|w_{t+1}\|^{2}\mid \mathcal{F}_{t}\right]<\infty, \;\;\; \|w_{t}\|=o(\varphi(t)),\;\;a.s.,
\end{equation}
where $\varphi(\cdot)$ is a deterministic positive, non-decreasing function and satisfies
\begin{equation}
\sup_{k}\varphi(e^{k+1})/\varphi(e^{k})<\infty.
\end{equation}
Let $\{f_{t}(k)\}, t, k=1, 2, \cdots$ is an $\mathcal{F}_{t}-$measurable random sequence. Then for $p_{n}=O\left(\lfloor\log^{a} n\rfloor\right), a>1$, the following property holds as $n\to \infty$:
\begin{equation}
\max_{1\leq k\leq p_{n}}\max_{1\leq i\leq n}\left\| \sum_{t=1}^{i}f_{t}(k)w_{t+1}\right\|=O(a_{n}\varphi(n)\log\log n),\;\;a.s.,
\end{equation}
where
\begin{equation}
a_{i}=\max_{1\leq k\leq p_{n}} g_{i}(k),\; g_{i}(k)=\left[\sum_{t=1}^{i}\|f_{t}(k)\|^{2}+1\right]^{\frac{1}{2}},\; g_{0}(k)=1.
\end{equation}
\end{lemma}

{\it Proof of Theorem $\ref{le1}$.}
We first introduce the following notations:
\begin{equation}\label{e511}
\begin{aligned}
\psi_{k}(n)&=F_{k+1}\left(c_{k}-\phi_{k}^{\top}\theta\right)-F_{k+1}\left(c_{k}-\phi_{k}^{\top}(n)\theta_{k}(n)\right),\\
\gamma_{k+1}&=s_{k+1}-F_{k+1}\left(c_{k}-\phi_{k}^{\top}\theta\right),\\
\delta_{k}(n)&=\phi_{k}^{\top}\theta-\phi_{k}^{\top}(n)\theta(n),\\
\tilde{\theta}_{k}(n)&=\theta(n)-\hat{\theta}_{k}(n).
\end{aligned}
\end{equation}
By $(\ref{e511})$, it can be easily seen that
\begin{equation}
\left|\psi_{k}(n)\right|\leq 1,\;\; \left|\gamma_{k}\right|\leq 1\;\;1\leq k\leq n,
\end{equation}
and
\begin{equation}
\mathbb{E}\left[\gamma_{k+1}\mid \mathcal{F}_{k}\right]=0,
\end{equation}
which means $\left\{\gamma_{k}, \mathcal{F}_{k}\right\}$ is a martingale difference sequence.

Next, we will consider the following stochastic  Lyapunov function: $$V_{k}(n)=\tilde{\theta}_{k}^{\top}(n)P_{k}^{-1}(n)\tilde{\theta}_{k}(n),\;\; 1\leq k\leq n. $$ 
According to Assumption \ref{ass1} and the definition of $h()\cdot$ in $(\ref{e8})$, we have $\lim\limits_{t\to 0}h(t)=\infty.$ Thus, by the definition of $d_{k}$ in $(\ref{e21})$-$(\ref{22})$, we obtain
\begin{equation}
\lim_{k\to \infty}d_{k}=\infty.
\end{equation}
Therefore, there exists a positive integer $n_{0}$, such that for all $n\geq n_{0}$,  
 \begin{equation}\label{e54}
 \theta\in D_{n_{0}}(n),\;\; n\geq n_{0}.
 \end{equation}
Hence, by $(\ref{e54})$ and the properties of the projection operator (see e.g., Lemma 1 in \cite{ZZ2022}), for each $ k\geq n_{0}$,
\begin{equation}\label{e555}
\begin{aligned}
V_{k}(n)\leq& \left[\tilde{\theta}_{k-1}(n)-a_{k-1}(n)\beta_{k-1}P_{k-1}(n)\phi_{k-1}(n)e_{k}(n)\right]^{\top}\cdot\\
&P_{k}^{-1}(n)\left[\tilde{\theta}_{k-1}^{\top}(n)-a_{k-1}(n)\beta_{k-1}P_{k-1}(n)\phi_{k-1}(n)e_{k}(n)\right].
\end{aligned}
\end{equation}
Furthermore, note that by $(\ref{9})$ and the well-known matrix inversion formula, we have
\begin{equation}\label{e5566}
P_{k+1}^{-1}(n)=P_{k}^{-1}(n)+\beta_{k}^{2}\phi_{k}(n)\phi_{k}^{\top}(n),\;\;0\leq k\leq n-1,
\end{equation}
thus we have
\begin{equation}\label{28}
	\begin{aligned}
		&a_{k-1}(n)P_{k}^{-1}(n)P_{k-1}(n)\phi_{k-1}(n)\\
		=&a_{k-1}(n)\left(I+\beta_{k-1}^{2}\phi_{k-1}(n)\phi_{k-1}^{\top}(n)P_{k-1}(n)\right)\phi_{k-1}(n)\\
		=&a_{k-1}(n)\phi_{k-1}(n)\left(1+\beta_{k-1}^{2}\phi_{k-1}^{\top}(n)P_{k-1}(n)\phi_{k-1}(n)\right)\\
		=&\phi_{k-1}(n), 
\end{aligned}
\end{equation}
and 
\begin{equation}\label{xx}
\begin{aligned}
	&\tilde{\theta}_{k-1}^{\top}(n)P_{k}^{-1}(n)\tilde{\theta}_{k-1}(n)\\
	=&\tilde{\theta}_{k-1}^{\top}(n)P_{k-1}^{-1}(n)\tilde{\theta}_{k-1}(n)\\
	&+\beta_{k-1}^{2}\tilde{\theta}_{k-1}^{\top}(n)\phi_{k-1}(n)\phi_{k-1}^{\top}(n)\tilde{\theta}_{k-1}(n).
	\end{aligned}
\end{equation}
By $(\ref{e555})$, $(\ref{28})$ and $(\ref{xx})$, we can obtain
\begin{equation}\label{27}
	\begin{aligned}
		V_{k}(n)
		\leq&V_{k-1}(n)
		+\beta_{k-1}^{2}\tilde{\theta}_{k-1}^{\top}(n)\phi_{k-1}(n)\phi_{k-1}^{\top}(n)\tilde{\theta}_{k-1}(n)\\
		&-2\beta_{k-1}\tilde{\theta}_{k-1}^{\top}(n)\phi_{k-1}(n)\psi_{k-1}(n)\\
		&+a_{k-1}(n)\beta_{k-1}^{2}\phi_{k-1}^{\top}(n)P_{k-1}(n)\phi_{k-1}(n)\psi_{k-1}^{2}(n)\\
		&+2a_{k-1}(n)\beta_{k-1}^{2}\phi_{k-1}^{\top}(n)P_{k-1}(n)\phi_{k-1}(n)\gamma_{k}\\
		&-2\beta_{k-1}\phi_{k-1}^{\top}(n)\tilde{\theta}_{k-1}(n)\gamma_{k}\\
		&+a_{k-1}(n)\beta_{k-1}^{2}\phi_{k-1}^{\top}(n)P_{k-1}(n)\phi_{k-1}(n)\gamma_{k}^{2}.
	\end{aligned}
\end{equation}
By the definition of $\beta_{k}$ in $(\ref{9})$, we have
\begin{equation}\label{29}
	\begin{aligned}
		&\beta_{k-1}\tilde{\theta}_{k-1}^{\top}(n)\phi_{k-1}(n)\psi_{k-1}(n)\\
		\geq &\beta_{k-1}^{2}\tilde{\theta}_{k-1}^{\top}(n)\phi_{k-1}(n)\left(\tilde{\theta}_{k-1}^{\top}(n)\phi_{k-1}(n)+\delta_{k-1}(n)\right)\\
		\geq&\beta_{k-1}^{2}\left(\tilde{\theta}_{k-1}^{\top}(n)\phi_{k-1}(n)+\delta_{k-1}(n)\right)^{2}\\
		&-\beta_{k-1}\delta_{k-1}(n)\psi_{k-1}(n).
	\end{aligned}
\end{equation}
By the fact that $\left|\psi_{k-1}(n)\right|\leq 1$, we have
\begin{equation}\label{31}
	\begin{aligned}
		 &a_{k-1}(n)\beta_{k-1}^{2}\phi_{k-1}^{\top}(n)P_{k-1}(n)\phi_{k-1}(n)\psi_{k-1}^{2}(n)\\
		 \leq &a_{k-1}(n)\beta_{k-1}^{2}\phi_{k-1}^{\top}(n)P_{k-1}(n)\phi_{k-1}(n).
	\end{aligned}
\end{equation}
Now, substituting $(\ref{29})$ and  $(\ref{31})$ into $(\ref{27})$, and
summing up both sides of $\left(\ref{27} \right)$ from $k=n_{0}$ to $n$, we obtain
\begin{equation}\label{32}
	\begin{aligned}
		V_{n}(n)\leq &V_{n_{0}}(n)+\sum_{k=n_{0}}^{n-1}\beta_{k}^{2}\left[\tilde{\theta}_{k}^{\top}(n)\phi_{k}(n)\right]^{2}\\
		&-2\sum_{k=n_{0}}^{n-1}\beta_{k}^{2}\left(\tilde{\theta}_{k}^{\top}(n)\phi_{k}(n)+\delta_{k}(n)\right)^{2}\\
		&+2\sum_{k=n_{0}}^{n-1}\beta_{k}\delta_{k}(n)\psi_{k}(n)\\
		&+\sum_{k=n_{0}}^{n-1}a_{n-1}(n)\beta_{n-1}^{2}\phi_{n-1}^{\top}(n)P_{n-1}(n)\phi_{n-1}(n)\\
		&-2\sum_{k=n_{0}}^{n-1}\beta_{k}\phi_{k}^{\top}(n)\tilde{\theta}_{k}(n)\gamma_{k+1}\\
		&+2\sum_{k=n_{0}}^{n-1}a_{k}(n)\beta_{k}^{2}\psi_{k}(n)\phi_{k}^{\top}(n)P_{k}(n)\phi_{k}(n)\gamma_{k+1}\\
		&+\sum_{k=n_{0}}^{n-1}a_{k}(n)\beta_{k}^{2}\phi_{k}^{\top}(n)P_{k}(n)\phi_{k}(n)\gamma_{k+1}^{2}.
	\end{aligned}
\end{equation}
We now analyze the right-hand side (RHS) of $(\ref{32})$ term by term. For the third term of $(\ref{32})$, by the definition of $\delta_{k}(n)$ in $(\ref{e511})$, the properties of $\{G_{i}\}$ in $(\ref{pro})$, and the condition on $\{u_{i}\}$ in $(\ref{ee25})$, we have
\begin{equation}\label{e655}
\begin{aligned}
\sum_{k=n_{0}}^{n-1}\delta_{k}^{2}(n)
=&O\left(\sum_{k=n_{0}}^{n-1}\left[\phi_{k}^{\top}\theta-\phi_{k}^{\top}(n)\theta(n)\right]^{2}\right)\\
=&O\left(\sum_{k=n_{0}}^{n-1}\left(\sum_{i=p_{n}+1}^{\infty}\|G_{i}\|\right)^{2}\max_{p_{n}\leq i \leq k}u_{k-i}^{2}\right)\\
=&O\left(n^{3+2\delta}e^{-rp_{n}}\right)=O(1),\;\;a.s.
\end{aligned}
\end{equation}
Notice that
\begin{equation}\label{e66}
\begin{aligned}
&\sum_{k=n_{0}}^{n-1}\beta_{k}^{2}\tilde{\theta}_{k}^{\top}(n)\phi_{k}(n)\delta_{k}(n)\\
\leq &\frac{1}{8}\sum_{k=n_{0}}^{n-1}\beta_{k}^{2}(n)\left[\tilde{\theta}_{k}^{\top}(n)\phi_{k}(n)\right]^{2}+4\sum_{k=n_{0}}^{n-1}\beta_{k}^{2}(n)\delta_{k}^{2}(n),\;\;a.s.
\end{aligned}
\end{equation}
Thus, by $(\ref{e655})$ and $(\ref{e66})$, we have
\begin{equation}\label{e134}
\begin{aligned}
&-2\sum_{k=n_{0}}^{n-1}\beta_{k}^{2}\left(\tilde{\theta}_{k}^{\top}(n)\phi_{k}(n)+\delta_{k}(n)\right)^{2}\\
\leq &-\frac{3}{2}\sum_{k=n_{0}}^{n-1}\beta_{k}^{2}\left[\tilde{\theta}_{k}^{\top}(n)\phi_{k}(n)\right]^{2}+18\sum_{k=n_{0}}^{n-1}\beta_{k}^{2}\delta_{k}^{2}(n)\\
= &-\frac{3}{2}\sum_{k=n_{0}}^{n-1}\beta_{k}^{2}\left[\tilde{\theta}_{k}^{\top}(n)\phi_{k}(n)\right]^{2}+O(1),\;\;a.s.
\end{aligned}
\end{equation}
For the fourth term of $(\ref{32})$, we obtain
\begin{equation}\label{e135}
\begin{aligned}
&\sum_{k=n_{0}}^{n-1}\beta_{k}\psi_{k}(n)\delta_{k}(n)\\
=&O\left(\sum_{k=n_{0}}^{n-1}|\delta_{k}(n)|\right)=O\left(\sqrt{n\sum_{k=n_{0}}^{n-1}\delta_{k}^{2}(n)}\right)\\
=&O\left(n^{2+\delta}e^{-rp_{n}}\right)=O(1),\;\;a.s.
\end{aligned}
\end{equation}
For the fifth term of $(\ref{32})$, by Lemma \ref{lem3} in Appendix, we have
\begin{equation}\label{a136}
\begin{aligned}
&\sum_{k=n_{0}}^{n-1}a_{k}(n)\beta_{k}^{2}\phi_{k}^{\top}(n)P_{k}(n)\phi_{k}(n)\\
=&O\left(\log^{+}\{\det\left(P_{n}^{-1}(n)\right)\}+1\right),\;\;a.s.
\end{aligned}
\end{equation}
Note that by $(\ref{e5566})$,
\begin{equation}\label{a137}
\begin{aligned}
&\log^{+}\left(\det\left(P_{n}^{-1}(n)\right)\right)\\
\leq&p_{n}\log^{+}\left(tr\left[P_{0}^{-1}(n)+\sum_{k=0}^{n-1}\beta_{k}^{2}\phi_{k}(n)\phi_{k}^{\top}(n)\right]\right)\\
=&p_{n}\log^{+}\left(tr\left[P_{0}^{-1}(n)\right]+\sum_{k=0}^{n-1}\beta_{k}^{2}\sum_{i=0}^{p_{n}+1}u_{k-i}^{2}\right)\\
=&O\left(p_{n}\log \left(n^{3+3\delta}\right)\right)=O\left(p_{n}\log n\right),\;\;a.s.
\end{aligned}
\end{equation}
Hence, by $(\ref{a136})$ and $(\ref{a137})$, we have
\begin{equation}\label{e138}
\sum_{k=n_{0}}^{n-1}a_{k}(n)\beta_{k}^{2}\phi_{k}^{\top}(n)P_{k}(n)\phi_{k}(n)=O\left(p_{n}\log n\right),\;\;a.s.
\end{equation}
Furthermore, we analyze the last three terms in $(\ref{32})$ which are related to the martingale difference sequence $\left\{\gamma_{n}, \mathcal{F}_{n}\right\}$ by using the double array martingale limit theory. Since $\left|\gamma_{k}\right|\leq 1$, we have $\|\gamma_{k}\|=o(\log \log (k+e)), a.s.$ Thus, by Lemma \ref{gannals}, we have
\begin{equation}
\begin{aligned}
&\sum_{k=n_{0}}^{n-1}\beta_{k}\phi_{k}^{\top}(n)\tilde{\theta}_{k}(n)\gamma_{k+1}\\
=&O\left(\left(\sum_{k=n_{0}}^{n-1}\beta_{k}^{2}\left[\tilde{\theta}_{k}^{\top}(n)\phi_{k}(n)\right]^{2}\right)^{\frac{1}{2}}(\log\log n)^{2}\right),\\
=&o\left(\sum_{k=n_{0}}^{n-1}\beta_{k}^{2}\left[\tilde{\theta}_{k}^{\top}(n)\phi_{k}(n)\right]^{2}\right)+o\left(\left(\log\log n\right)^{5}\right),\;\;a.s.
\end{aligned}
\end{equation} 
 For the seventh term of $(\ref{32})$, by Lemma \ref{gannals} and $(\ref{a136})$-$(\ref{a137})$, we know that
 \begin{equation}
 \begin{aligned}
 &\sum_{k=n_{0}}^{n-1}a_{k}(n)\beta_{k}^{2}\psi_{k}(n)\phi_{k}^{\top}(n)P_{k}(n)\phi_{k}(n)\gamma_{k+1}\\
 =&o\left(\sum_{k=n_{0}}^{n-1}\left[a_{k}(n)\beta_{k}^{2}(n)\phi_{k}^{\top}(n)P_{k}(n)\phi_{k}(n)\right]^{2}\right)\\
 &+o\left(\left(\log\log n\right)^{5}\right)\\
 =&o\left(p_{n}\log n\right)+o\left(\left(\log\log n\right)^{5}\right),\;\;a.s.
 \end{aligned}
 \end{equation}
For the eighth term of $(\ref{32})$, by the fact that $|\gamma_{k+1}|\leq 1$ and $(\ref{e138})$, we have
\begin{equation}\label{e141}
\begin{aligned}
 \sum_{k=n_{0}}^{n-1}a_{k}(n)\beta_{k}^{2}\phi_{k}^{\top}(n)P_{k}(n)\phi_{k}(n)\gamma_{k+1}^{2}
 =O\left(p_{n}\log n\right),\;\;a.s.
 \end{aligned}
 \end{equation}
Substituting $(\ref{e134})$, $(\ref{e135})$, $(\ref{e138})$-$(\ref{e141})$ into $(\ref{32})$, we finally obtain
\begin{equation}\label{e147}
\begin{aligned}
&\tilde{\theta}_{n}^{\top}(n)P_{n}^{-1}(n)\tilde{\theta}_{n}(n)+\frac{1}{2}\sum_{k=n_{0}}^{n-1}\beta_{k}^{2}\left[\tilde{\theta}_{k}^{\top}(n)\phi_{k}(n)\right]^{2}\\
=&O\left(p_{n}\log n\right),\;\;a.s.
\end{aligned}
\end{equation}

Furthermore, by the definition of $\beta_{n}$ and $g(\cdot)$ in $(\ref{9})$ and $(\ref{e7})$, for each $n\geq 1$, we have 
\begin{equation}\label{148}
\beta_{n}\geq g(0)\alpha_{n}^{-2}.
\end{equation}
Thus, by $(\ref{e5566})$ and $(\ref{148})$, we have
\begin{equation}\label{500i}
\begin{aligned}
\lambda_{\min}\left\{P_{n}^{-1}(n)\right\}
&\geq \lambda_{\min}\left\{\sum_{t=1}^{n}\beta_{t}^{2}\phi_{t}(n)\phi_{t}^{\top}(n)\right\}\\
&\geq g(0)\alpha_{n}^{-2} \lambda_{\min}\left\{\sum_{t=1}^{n}\phi_{t}(n)\phi_{t}^{\top}(n)\right\}.
\end{aligned}
\end{equation}
Hence, $(\ref{6})$ is obtained by $(\ref{e147})$ and $(\ref{500i})$.

For the proof of $(\ref{e27})$, drawing inspiration from the analysis of $ARX(\infty)$ approximation in \cite{g2023}. Let $s_{i}=\lfloor e^{\frac{i}{1+a}}\rfloor$ for each $i\geq 0$. Then for each $s_{i}+1\leq k \leq s_{i+1}$, we have
\begin{equation}\label{e143}
p_{k}=p_{s_{i+1}}.
\end{equation}
Notice that $\{s_{i}, i\geq 0\}$ is a increasing sequence, thus for each $n\geq 1$, there exsits a positive integer $j_{n}$ such that $s_{j_{n}}\leq n < s_{j_{n}+1}$. Moreover, we can easily verify that 
\begin{equation}\label{e144}
j_{n}\leq p_{n+1}.
\end{equation}
Therefore, by $(\ref{e143})$ and $(\ref{e144})$, we can obtain
\begin{equation}\label{152}
\begin{aligned}
&\sum_{k=n_{0}}^{n-1}\beta_{k}^{2}\left[\tilde{\theta}_{k}^{\top}(k)\phi_{k}(k)\right]^{2}\\
\leq&\sum_{i=0}^{j_{n}}\sum_{k=s_{i}+1}^{s_{i+1}}\beta_{k}^{2}\left[\tilde{\theta}_{k}^{\top}(k)\phi_{k}(k)\right]^{2}\\
=&\sum_{i=0}^{j_{n}}\sum_{k=s_{i}+1}^{s_{i+1}}\beta_{k}^{2}\left[\tilde{\theta}_{k}^{\top}(s_{i+1})\phi_{k}(s_{i+1})\right]^{2}\\
=&O\left(\sum_{i=0}^{j_{n}}p_{s_{i+1}}\log p_{s_{i+1}}\right)=O\left(p_{n}^{2}\log n\right),\;\;a.s.
\end{aligned}
\end{equation}
Thus, by $(\ref{148})$ and $(\ref{152})$, we have
\begin{equation}\label{f}
\sum_{k=n_{0}}^{n-1}\left[\tilde{\theta}_{k}^{\top}(k)\phi_{k}(k)\right]^{2}=O\left(\alpha_{n}^{2}p_{n}^{2}\log n\right),\;\;a.s.
\end{equation}
Furthermore, notice that
\begin{equation}\label{ff}
\sum_{k=0}^{n}\delta_{k}^{2}(k)=\sum_{t=0}^{n}\lambda^{2p_{k}}k^{2+2\delta}=O\left(1\right).
\end{equation}
Thus, by $(\ref{f})$ and $(\ref{ff})$, we finally obtain
\begin{equation}
\sum_{k=0}^{n}R_{k}=\sum_{k=0}^{n}\left[\tilde{\theta}_{k}^{\top}(k)\phi_{k}(k)+\delta_{k}(k)\right]^{2}=O\left(\alpha_{n}^{2}p_{n}^{2}\log n\right),\;\;a.s.
\end{equation}

\subsection{Proof of Theorem \ref{thm2}}
Before presenting the proof of Theorem 2, we first prove two key lemmas.

Set
$$\psi_{t}(n)=\left[v_{t}, v_{t-1},\cdots,v_{t-p_{n}}\right]^{\top},\;\; 1\leq t \leq n,$$
where $\{v_{t}\}$ is defined as in $(\ref{e388})$. Then, using double array martingale theory, we can obtain the following lemma regarding the minimum eigenvalue of the information matrix.
\begin{lemma}\label{lem33}
For each $\frac{1}{2}<\eta<1$, there exists constants $c>0$ and $N\geq 1$ such that for each $n\geq N$, we have
\begin{equation}
\lambda_{\min}\left\{\sum_{t=n-n^{\eta}}^{n}\psi_{t}(n)\psi_{t}^{\top}(n)\right\}\geq cn^{\eta}.
\end{equation}
\end{lemma}
\begin{proof}
For each $n\geq 1$ and any vector $A_{n}=\left[a_{1},\cdots, a_{p_{n}}\right]^{\top}$ with $\|A_{n}\|=1$, we have
\begin{equation}\label{400}
\begin{aligned}
&A_{n}^{\top}\sum_{t=n-n^{\eta}}^{n}\psi_{t}(n)\psi_{t}^{\top}(n)A_{n}\\
=&\sum_{t=n-n^{\eta}}^{n}\left(a_{1}v_{t}+\cdots a_{p_{n}}v_{t-p_{n}+1}\right)^{2}\\
=&\sum_{t=n-n^{\eta}}^{n}\left(a_{1}^{2}v_{t}^{2}+\cdots a_{p_{n}}^{2}v^{2}_{t-p_{n}+1}\right)\\
&+2\sum_{t=n-n^{\eta}}^{n}\sum_{0\leq s<r\leq p_{n}}a_{s+1}a_{r+1}v_{t-s}v_{t-r}.
\end{aligned}
\end{equation}
From Lemma \ref{lem2} in appendix, we have
\begin{equation}\label{411}
\begin{aligned}
&\max_{1\leq s \leq p_{n}}\max_{r< s \leq p_{n}}\left\|\sum_{t=n-n^{\eta}}^{n}v_{t-s}v_{t-r}\right\|\\
=&\max_{1\leq s \leq p_{n}}\max_{r< s \leq p_{n}}\left\|\sum_{t=1}^{n}v_{t-s}v_{t-r}-\sum_{t=1}^{n-n^{\eta}}v_{t-s}v_{t-r}\right\|\\
\leq &\max_{1\leq s \leq p_{n}}\max_{r< s \leq p_{n}}\left\|\sum_{t=1}^{n}v_{t-s}v_{t-r}\right\|\\
&+\max_{1\leq s \leq p_{n}}\max_{r< s \leq p_{n}}\left\|\sum_{t=1}^{n-n^{\eta}}v_{t-s}v_{t-r}\right\|\\
\leq &C_{1}\sqrt{n\log\log n},
\end{aligned}
\end{equation}
where $C_{1}$ does not depend on $n$.
Besides, we have 
\begin{equation}\label{422}
\sum_{t=n-n^{\delta}}^{n}\left(a_{1}^{2}\mathbb{E}\left[v_{t}^{2}\right]+\cdots a_{p_{n}}^{2}\mathbb{E}\left[v^{2}_{t-p_{n}+1}\right]\right)\geq n^{\eta},
\end{equation}
and
\begin{equation}\label{433}
\max_{0\leq s\leq p_{n}-1}a_{s}^{2}\left\|\sum_{t=n-n^{\eta}}^{n}\left(v_{t-s}^{2}-\mathbb{E}\left[v_{t-s}^{2}\right]\right)\right\|\leq C_{2}\sqrt{n\log\log n},
\end{equation}
where $C_{2}$ does not depend on $n$. From $(\ref{400})$, $(\ref{411})$, $(\ref{422})$ and $(\ref{433})$, we obtain that
\begin{equation}
\begin{aligned}
&\frac{1}{n^{\eta}}\lambda_{\min}\left\{\sum_{t=n-n^{\eta}}^{n}\psi_{t}(n)\psi_{t}^{\top}(n)\right\}\\
\geq &1-(C_{1}+C_{2})\frac{p_{n}\sqrt{n\log\log n}}{n^{\eta}},
\end{aligned}
\end{equation}
which deduces Lemma \ref{lem33}.
\end{proof}

The following lemma provides the stability results for the current IIR system.

Given a series $Q(z)=1+\sum\limits_{i=1}^{\infty}q_{i}z^{i}$ with 
\begin{equation}\label{e411}
\sum\limits_{i=1}^{\infty}\left|q_{i}\right|<\infty,\;\; Q(z)\not=0,\;\;\forall |z|\leq 1,
\end{equation}
and a sequence of polynomials $$Q_{n}(z)=1+\sum\limits_{i=1}^{p_{n}}q_{n,i}z^{i},\;\;n\geq 1,$$ where $p_{n}$ is defined as in $(\ref{pn})$, the following lemma holds: 
\begin{lemma}\label{lem22}
Suppose for $n_{0}\leq n \leq N$, we have
\begin{equation}\label{ee42}
Q_{n}(z)u_{n}=\xi_{n},
\end{equation}
where $z$ is the backwards-shift operator, $\{\xi_{n}\}$ is any bounded sequence, $n_{0}$ and $N$ are any positive integers. 
Then there exists a positive constant $\Delta$ depending only on $\{q_{i}, i\geq 1\}$, such that if the condition
\begin{equation}
\sup_{n_{0}\leq n \leq N}\sum\limits_{i=1}^{\infty}\left|q_{i}-q_{n,i}\right|^{2}\leq \Delta^{2}
\end{equation}
holds, where $q_{n,i}\overset{\triangle}{=}0$ for all $i\geq p_{n}$, then it follows that
\begin{equation}
|u_{N}|\leq c_{1}\max_{n_{0}-p_{n_{0}}\leq i\leq n_{0}}|u_{i}|+c_{2}, 
\end{equation}
where $c_{1}$ and $c_{2}$ are positive constants which do not depend on $n_{0}$ and $N$.
\end{lemma}
\begin{proof}
For each positive integer $k$, set
\begin{equation}
A_{k}=\left[\begin{matrix} 
-q_{1} & \cdots &-q_{k} & 0\\ 
1 & \cdots & 0& 0\\
\cdots & \cdots &\cdots & \cdots \\
0 & \cdots &1 &0 
\end{matrix}\right].
\end{equation}
From $(\ref{e411})$, there exists a positive integer $k$ such that $\rho(A_{k})<1$ and
\begin{equation}
\sum\limits_{i=k+1}^{\infty}\left|q_{i}\right|<\frac{1}{4}-\frac{1}{4}\rho(A_{k}).
\end{equation}
We now define a matrix sequence $\{A_{n,k}, n\geq 1\}$ as follows:
\begin{equation}
A_{n,k}=\left[\begin{matrix} 
-q_{n,1} & \cdots &-q_{n,k} & 0\\ 
1 & \cdots & 0& 0\\
\cdots & \cdots &\cdots & \cdots \\
0 & \cdots &1 &0 
\end{matrix}\right].
\end{equation}
Then we can obtain 
\begin{equation}
\left\|A_{k}-A_{n,k}\right\|^{2}\leq \sum_{i=1}^{k}\left|q_{i}-q_{n,i}\right|^{2},
\end{equation}
where we set $q_{n,i}=0, \forall i> p_{n}$.
Thus, let $\Delta=\min\left\{\frac{1}{2}\|A_{k}\|, \frac{1}{12}-\frac{1}{12}\rho(A_{k})\right\}$, we have
\begin{equation}
\begin{aligned}
\sup_{n_{0}\leq n\leq N}\left\|A_{n,k}\right\|&\leq \left\|A_{k}\right\|+\Delta\leq \frac{3}{2}\left\|A_{k}\right\|,
\end{aligned}
\end{equation}
\begin{equation}
\begin{aligned}
\sup_{n_{0}\leq n\leq N}\rho(A_{n,k})\leq \rho(A_{k})+\Delta
\leq \frac{1}{4}+\frac{3}{4}\rho\left(A_{k}\right),
\end{aligned}
\end{equation}
and 
\begin{equation}\label{ee51}
\begin{aligned}
\sup\limits_{n_{0}\leq n\leq N}\sum\limits_{i=k+1}^{p_{n}}\left|q_{n, i}\right|\leq \sum\limits_{i=k+1}^{\infty}\left|q_{i}\right|+\Delta
\leq \frac{1}{3}-\frac{1}{3}\rho(A_{k}).
\end{aligned}
\end{equation}
Hence, let $\epsilon=\frac{1-\rho\left(A_{k}\right)}{6\|A_{k}\|}$, $M=(1+\frac{2}{\epsilon})^{k-1}\sqrt{k}$ and $\lambda=\frac{1}{2}+\frac{1}{2}\rho(A_{k})$ in Lemma \ref{lea} in appendix, we get
\begin{equation}\label{ee52}
\begin{aligned}
\sup_{n_{0}\leq n\leq N}\left\|A_{n,k}^{i}\right\|\leq &M\lambda^{i},\;\;\forall i\geq 0.
\end{aligned}
\end{equation}

For each $n\geq 1$, set
$$\bar{u}_{n}=\left\{
		\begin{array}{rcl}
			&u_{n}, \;\;\;           & n\geq n_{0}-p_{n_{0}}+1\\   			              &0,\;\;\; & n < n_{0}-p_{n_{0}}+1
		\end{array} \right.,$$
 and $x_{n}=\left[\bar{u}_{n}, \bar{u}_{n-1},\cdots, \bar{u}_{n-k+1}\right]^{\top}$, $\epsilon_{n}=\left[\xi_{n}, 0,\cdots, 0\right]^{\top}$. 
 Besides, for each $n\geq 1$ and $i\geq 1$, set
$$g_{n,i}=
\left\{
		\begin{array}{rcl}
			&q_{n,i},\;\;\;           & 1\leq i\leq p_{n}\\   			              &0,\;\;\; & i > p_{n}
		\end{array} \right.,$$
		and
$$h_{n,i}=\left[\begin{matrix} 
-g_{n,n-i} & 0 &\cdots & 0\\ 
0 & 0 &  \cdots  & 0\\
\cdots & \cdots &\cdots & \cdots \\
0 & 0 &\cdots &0 
\end{matrix}\right].$$ 
Then by $(\ref{ee42})$, we have
\begin{equation}
x_{n}=A_{n,k}x_{n-1}+\sum_{i=1}^{n-k-1}h_{n,i}x_{i}+\epsilon_{n},\;\; n\geq n_{0}.
\end{equation}
Define
\begin{equation}
\Phi(n+1, i)=A_{n+1,k}\Phi(n, i), \Phi(i, i)=I, \forall n\geq i\geq 1.
\end{equation}
We now prove that for each $n\geq n_{0}+1$,
\begin{equation}\label{355}
\begin{aligned}
x_{n}
=\Phi(n, n_{0})x_{n_{0}}+\sum_{i=n_{0}+1}^{n}\Phi(n,i)\epsilon_{i}+\sum_{i=1}^{n-k-1}l_{n,i}x_{i}.
\end{aligned}
\end{equation}
where the sequence $\{l_{n,i}, 1\leq i\leq n\}$ is defined as follows:
\begin{equation}\label{111}
l_{n,i}=
\left\{
		\begin{array}{rcl}
			&\sum\limits_{j=k+i+1}^{n}\Phi(n,j)h_{j,i},& n_{0}-k\leq i\leq n-k-1\\   			              &\sum\limits_{j=n_{0}+1}^{n}\Phi(n,j)h_{j,i},& 1\leq i \leq n_{0}-k-1
		\end{array} \right..
		\end{equation}
Indeed, for $n=n_{0}+1$, we can easily verify that
\begin{equation}
l_{n_{0}+1,i}=h_{n_{0}+1,i},\;\;1\leq i\leq n-k-1,
\end{equation}
thus $(\ref{355})$ is satisfied. Assume that $(\ref{355})$ holds for some $n\geq n_{0}+1$, we need to show that it also holds for $n+1$. For this, we have
\begin{equation}
\begin{aligned}
&x_{n+1}=A_{n+1,k}x_{n}+\sum_{i=1}^{n-k}h_{n+1,i}x_{i}+\epsilon_{n+1}\\
=&A_{n+1,k}\Phi(n, n_{0})x_{n_{0}}+A_{n+1,k}\sum_{i=n_{0}+1}^{n}\Phi(n,i)\epsilon_{i}\\
&+A_{n+1,k}\sum_{i=1}^{n-k-1}l_{n,i}x_{i}+\sum_{i=1}^{n-k}h_{n+1,i}x_{i}+\epsilon_{n+1}\\
=&\Phi(n+1,n_{0})x_{n_{0}}+\sum_{i=n_{0}+1}^{n+1}\Phi(n+1,i)\epsilon_{i}+h_{n+1,n-k} x_{n-k}\\
&+\sum_{i=1}^{n-k-1}\left(A_{n+1,k}l_{n,i}+h_{n+1,i}\right)x_{i}.
\end{aligned}
\end{equation}
Notice that by $(\ref{111})$, we have $l_{n,n-k-1}=h_{n,n-k-1}$. Besides, for each $n_{0}-k\leq i \leq n-k$, we have
\begin{equation}
\begin{aligned}
A_{n,k}l_{n-1,i}+h_{n,i}=&A_{n,k}\sum\limits_{j=k+i+1}^{n-1}\Phi(n-1,j)h_{j,i}+h_{n,i}\\
=&\sum\limits_{j=k+i+1}^{n-1}\Phi(n,j)h_{j,i}+h_{n,i}=l_{n,i}.
\end{aligned}
\end{equation}
Meanwhile, for each $1\leq i \leq n_{0}-k-1$, we have
\begin{equation}
\begin{aligned}
A_{n,k}l_{n-1,i}+h_{n,i}
=\sum\limits_{j=n_{0}+1}^{n}\Phi(n,j)h_{j,i}=l_{n,i}.
\end{aligned}
\end{equation}
Hence, $(\ref{355})$ is derived for $n+1$, completing the inductive step.

According to $(\ref{ee52})$ and $(\ref{355})$, for each $n\geq n_{0}+1$, we have
\begin{equation}\label{e49}
\begin{aligned}
\|x_{n}\| \leq M\lambda^{n-n_{0}}\|x_{n_{0}}\|+\sum_{i=n_{0}+1}^{n}\lambda^{n-i}\|\epsilon_{i}\|+\sum_{i=1}^{n-k-1}\|l_{n,i}\|\|x_{i}\|.
\end{aligned}
\end{equation}
By the definition of $l_{n,i}$ and (\ref{ee51}), we get
\begin{equation}\label{e501}
\setlength\belowdisplayskip{3pt}
\begin{aligned}
&\sum_{i=1}^{n-k-1}\|l_{n,i}\|\\
\leq & \sum_{i=1}^{n_{0}-k-1}\sum\limits_{j=n_{0}+1}^{n}\lambda^{n-j}\|h_{j,i}\|+\sum_{i=n_{0}-k}^{n-k-1}\sum\limits_{j=k+i+1}^{n}\lambda^{n-j}\|h_{j,i}\|\\
\leq &\sum_{i=1}^{n_{0}-k-1}\sum\limits_{j=n_{0}+1}^{n}\lambda^{n-j}\left|g_{j, j-i}\right|+\sum_{i=n_{0}-k+1}^{n-k-1}\sum\limits_{j=k+i+1}^{n}\lambda^{n-j}\left|g_{j, j-i}\right|\\
\leq &\sum_{j=n_{0}+1}^{n}\sum\limits_{i=1}^{j-k-1}\lambda^{n-j}\left|g_{j, j-i}\right|
\leq \sum_{j=n_{0}+1}^{n}\lambda^{n-j}\left(\sum\limits_{i=k+1}^{\infty}\left|q_{j, i}\right|\right)\\
\leq &\frac{1}{1-\lambda}\max_{n_{0}+1\leq j \leq n}\left(\sum\limits_{i=k+1}^{\infty}\left|q_{j, i}\right|\right)\leq \frac{2}{3}.
\end{aligned}
\end{equation}
Let
$y_{n}=\max\limits_{1\leq i\leq n}\|x_{i}\|.$
From $(\ref{e49})$ and $(\ref{e501})$,  for each $n\geq n_{0}+1$, we have
\begin{equation}
\begin{aligned}
y_{n}\leq &\frac{2}{3}y_{n-1}+M\lambda^{n-n_{0}}\|x_{n_{0}}\|+\sum_{i=n_{0}+1}^{n}\lambda^{n-i}\|\epsilon_{i}\|\\
 \leq &\left(\frac{2}{3}\right)^{n-n_{0}}y_{n_{0}}+\frac{M}{1-\lambda}\|x_{n_{0}}\|+c_{2}\\
 \leq & c_{1}\max_{n_{0}-p_{n_{0}}\leq i\leq n_{0}}\left|u_{i}\right|+c_{2}.
 \end{aligned}
\end{equation}
where $c_{i}, i=1, 2$, do not dependent on $n_{0}$ and $N$. 
Therefore, Lemma \ref{lem22} is obtained.
\end{proof}

{\it Proof of Theorems $\ref{thm2}$.}
Based on Lemma \ref{lem33} and  Lemma \ref{lem22}, the proof framework is based on the following three cases, with details provided elsewhere:
\begin{itemize}
\item[(1)] $\sigma_{k}<\infty$, $\tau_{k+1}=\infty$ for some $k$;
\item[(2)] $\tau_{k}<\infty$, $\sigma_{k}=\infty$ for some $k$;
\item[(3)] $\tau_{k}<\infty$, $\sigma_{k}<\infty$ for all $k$.
\end{itemize}

	\section{Numerical simulation}\label{se5}	
	In this section, we present a simulation example to illustrate the theoretical results obtained in the previous sections.
	
	Consider the following stochastic dynamic system under binary-valued output observations:
	\begin{equation}\label{eq11}
\left\{
\begin{aligned}
		  &A(z)y_{k+1}=B(z)u_{k}+v_{k+1}\\
		&s_{k+1}=I\left(y_{k+1} \geq c_{k}\right)     
		\end{aligned} \right., \;\;\;\;\;k\geq 0,
\end{equation}
where $A(z)=1-0.1z+0.5z^{2}, B(z)=1+0.5z-0.4z^{2}$.
The  noise sequence  $\left\{ v_{k}, k\geq 1\right\}$ is i.i.d. with standard normal distribution $\mathcal{N}(0,1)$. To estimate the parameter vector $\theta$ using Algorithm $\ref{alg1}$, we set the initial value as $\hat{\theta}_{0}(n)=[0, \cdots, 0]^{\top}$ and $P_{0}(n)=I$ for each $n\geq 1$.

The reference signals $\left\{y_{k}^{*}\right\}$ are generated by the following logistic map, which is in part as a discrete-time demographic model:
\begin{equation}\label{etar}
y_{k+1}^{*}=ry_{k}^{*}(1-y_{k}^{*}),\;\;k\geq 0,
\end{equation}
 where the initial value is set to $y_{0}^{*}=0.7$. It is well known that when different values of $r$ are selected within the range of $(0, 4)$, the system signal $\{y_{k}^{*}, k\geq 0\}$ exhibits different behaviors: 1) When $r\in (0,3]$, $y_{k}^{*}$ converges to a constant value; 2) When $r$ is within $\left(3, r^{*}\right)$ where $r^{*}=3.56994\cdots$, $y_{k}^{*}$ exhibit a periodic behavior; 3) When $r$
belongs to the interval $\left[r^{*}, 4\right)$, $y_{k}^{*}$ generally displays chaotic behavior. Below, we validate the tracking performance of the proposed adaptive controller for these three different types of reference signals, respectively.

\subsubsection*{\bf Case 1} Consider the case where $r=1.5$; in this scenario, the reference signals generated by $(\ref{etar})$ quickly converge to the value $\frac{1}{3}$. Fig.~\ref{fig_6} compares the tracking performance of our adaptive controller $(\ref{control})$-$(\ref{200})$ with the SA-based controller in \cite{wenxiao2023} for this type of reference signal, under both specifically designed binary observation thresholds $c_{k}$ and general observation thresholds $c_{k}$. a) When the binary observation thresholds $c_{k}$ are specifically designed to match the reference value of $\frac{1}{3}$, both controllers converge to the reference signal. b) When the binary observation thresholds $c_{k}$ can be set arbitrarily, we may take it as 0.8 for illustration. It can be observed that the system output under the SA-based controller in \cite{wenxiao2023} stabilizes at the observation threshold of $0.8$, failing to track the reference signal, while the system output under our adaptive controller can make the adaptive tracking error asymptotically to be zero. 

\begin{figure}[!t]
\centering
\subfloat[Specially designed binary observation thresholds]{
		\includegraphics[scale=0.45]{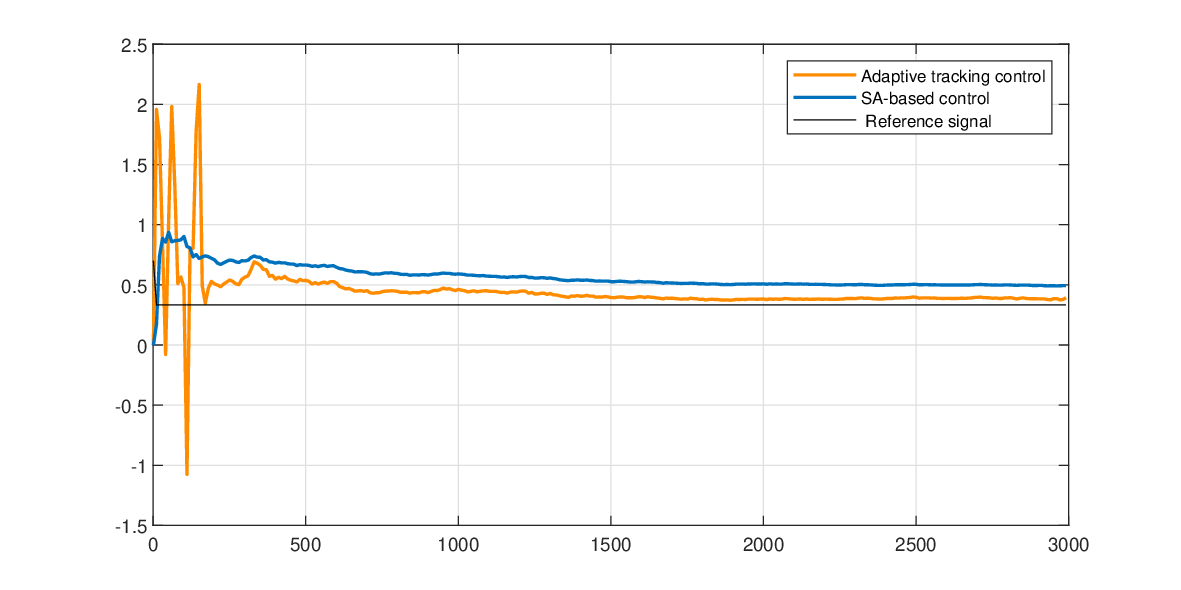}}
		\hfill
\subfloat[General binary observation thresholds]{
		\includegraphics[scale=0.45]{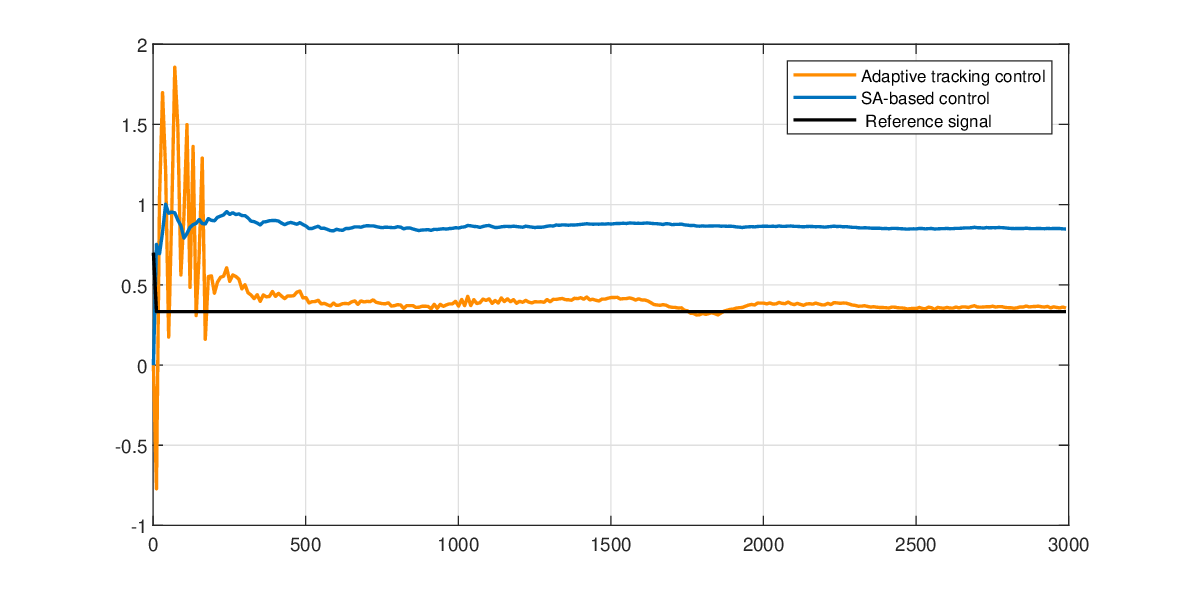}}
\caption{$r=1.5$: Adaptive tracking control for fixed reference signals. a) Specially designed binary observation thresholds b) General binary observation thresholds.}
\label{fig_6}
\end{figure}

\subsubsection*{\bf Case 2} Consider the case where $r=3.44$. In this situation, the reference signals generated by $(\ref{etar})$ will exhibit periodic oscillations between four values. In the experiment, the observation thresholds are set to $c_{k}=0$, and the adaptive controller $u_{k}$ is provided by $(\ref{control})$-$(\ref{200})$. Fig.~$\ref{fig_2}$ $(a)$-$(b)$ presents the tracking results: the upper subfigure shows the trajectories of the system outputs and the reference signals, from which we can see that the system outputs are able to track the variations in the reference signals; The lower subfigure clearly illustrates that the tracking error converges to zero.
\begin{figure}[!t]
\centering
\subfloat[Trajectories of reference signals and system outputs]{
		\includegraphics[scale=0.45]{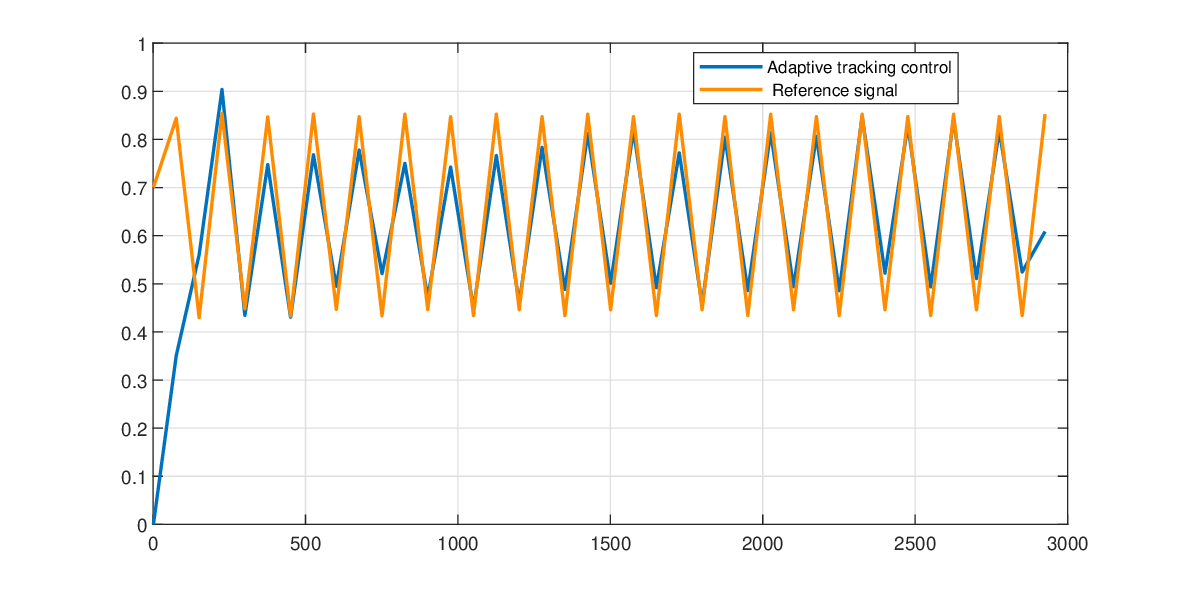}}
		\hfill
\subfloat[Tracking error]{
		\includegraphics[scale=0.45]{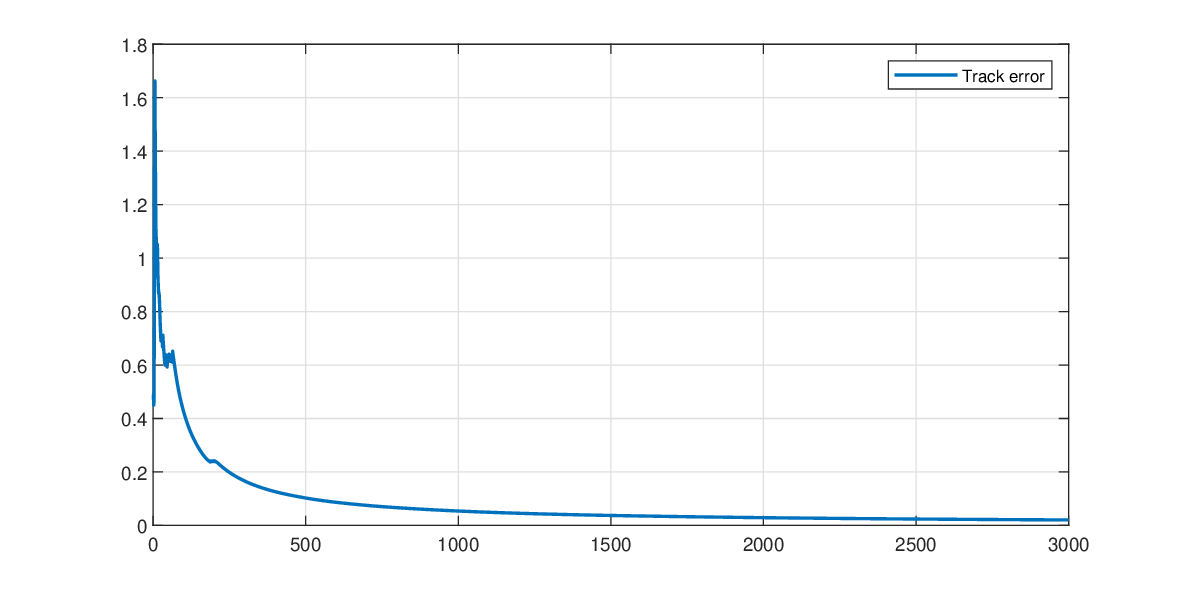}}
\caption{$r=3.44$: Adaptive tracking control for periodic reference signals.}
\label{fig_2}
\end{figure}

\subsubsection*{\bf Case 3} Consider the case where $r=3.8$.  In this scenario, the reference signal generated from $(\ref{etar})$ exhibits chaotic behavior. Fig.~$\ref{fig_3}$ $(a)$-$(b)$ show the trajectories of the reference signals, the system outputs under the adaptive controller $(\ref{control})$-$(\ref{200})$, as well as the tracking error. It can be seen that for these complicated reference signals, the adaptive tracking error can still converge to zero.

\begin{figure}[!t]
\centering
\subfloat[Trajectories of reference signals and system outputs]{
		\includegraphics[scale=0.45]{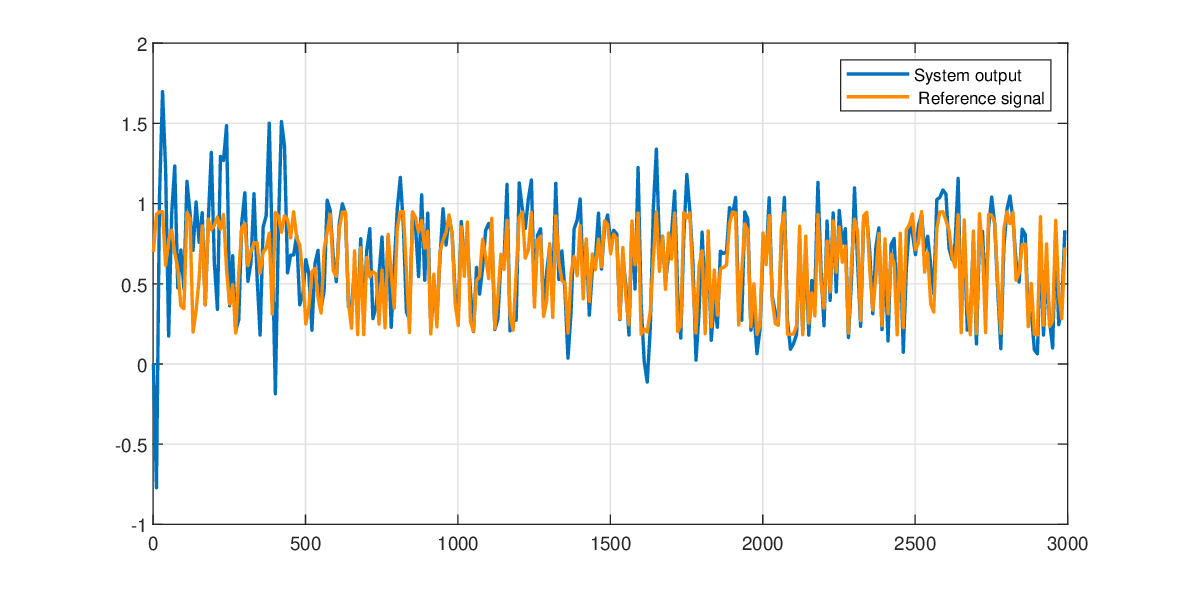}}
		\hfill
\subfloat[Tracking error]{
		\includegraphics[scale=0.45]{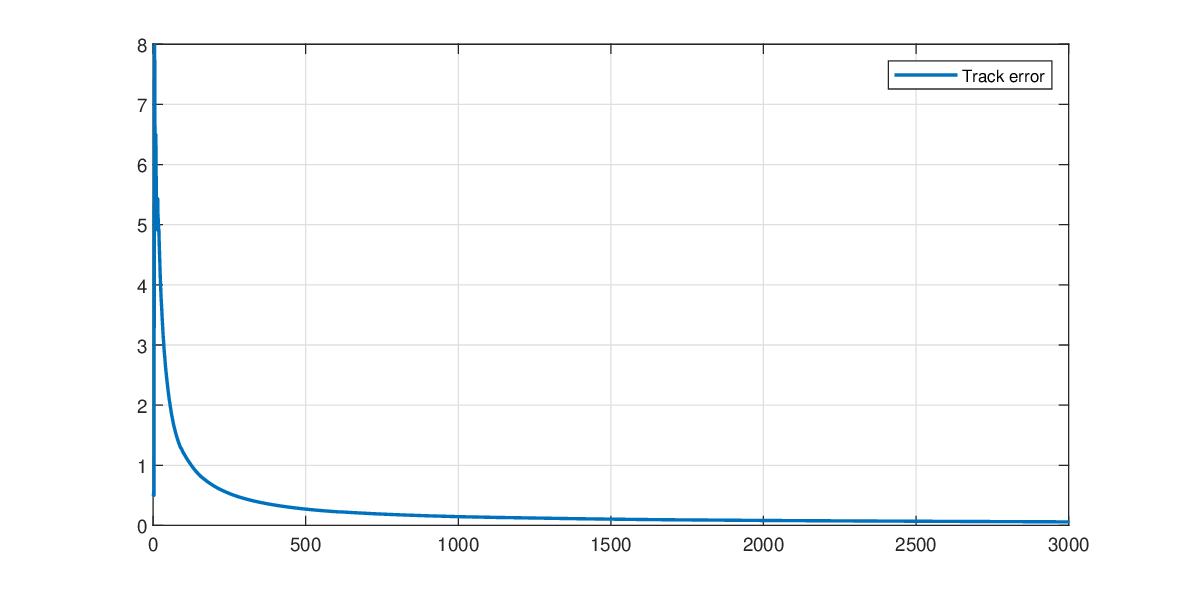}}
\caption{$r=3.8$: Adaptive tracking control for chaotic reference signals.}
\label{fig_3}
\end{figure}



	\section{Concluding remarks}\label{se6}	
In this paper, we have established an adaptive control theory for finite-dimensional linear systems with any given bounded reference signals under binary-valued output observations. As explained in the paper, one of the difficulties in establishing such a theory lies in the fact that the regression vectors containing the lagged output signals to be used in the construction of the traditional adaptive algorithms are not fully available. Another difficulty is how to sidestep or verify both the boundedness and the excitation conditions required in the convergence analysis of adaptive estimation algorithms in the existing related literature. To overcome these difficulties, we have used both the ideas of large model estimation and the theory of double array martingales in the design and analysis of adaptive algorithms, which makes it possible for us to establish for the first time an adaptive identification theory for linear systems with binary-valued output observations under non-PE conditions and possibly unbounded system signals. This further allows us to establish the global stability and asymptotic optimality of the closed-loop adaptive systems. We remark that various extensions of the current paper are possible, including the general case of quantized output observations. However, there are many interesting problems that need to be further investigated. At the moment, we do not know whether or not it is possible to establish such a theory without using the IIR-based design ideas and analysis techniques introduced in this paper. It would also be challenging problems to consider more complicated stochastic dynamical systems including multi-layer networks and other basic control problems such as linear-quadratic-adaptive control under binary-valued output observations.	
\begin{appendix}\label{ap}

\section{Proof of Theorem \ref{le1}}\label{BB} 
\begin{lemma}\label{lea}(\cite{g2020}).
Suppose $\rho\left(A\right)$ is the spectral radius of the matrix $A\in \mathbb{R}^{p\times p}$, then for any $\epsilon>0$, we have
\begin{equation}
\|A^{i}\|\leq M\lambda^{i}, \;\;\forall i\geq 0,
\end{equation}
where $M=\left(1+\frac{2}{\epsilon}\right)^{p-1}\sqrt{p},\;\;\lambda=\rho(A)+\epsilon\|A\|.$
\end{lemma}

\begin{lemma}\label{lem3} (\cite{g1990}). Let $\{f_{k}(n)\}, 1\leq k\leq n$ be a $\mathcal{F}_{k}$-measurable random vector sequence in $\mathbb{R}^{p_{n}}, p_{n}\geq 1,$ and $M_{k}(n)=I+\sum\limits_{j=1}^{k}f_{j}(n)f_{j}^{\top}(n), 1\leq k\leq n.$ Then as $n\rightarrow \infty,$
\begin{equation}
\sum_{k=1}^{n}f_{k}^{\top}(n)M_{k}^{-1}(n)f_{k}(n)=O\left(\log^{+}\{\det\left(M_{n}(n)\right)\}+1\right)
\end{equation}
where $det\left(M_{n}(n)\right)$ denotes the determinant of $M_{n}(n)$, and $log^{+}\{\cdot\}$ denotes the positive part of $\log\{\cdot\}$.
\end{lemma}

\begin{lemma}\label{lem2} (\cite{g1990}). Suppose that $\left\{w_{n}, \mathcal{F}_{n}\right\}$ is a martingale difference sequence satisfying	\begin{equation}\label{18}
		\sup_{n} 	\mathbb{E}[|w_{n+1}|^{2}\mid \mathcal{F}_{n}] < \infty,
		\end{equation}
		\begin{equation}
		 \sup_{n} 	\mathbb{E}[|w_{n+1}|^{4}(\log^{+}|w_{n}|)^{2+\delta}] < \infty,\;\;a.s.
	\end{equation}
and that $\{x_{n},\mathcal{F}_{n}\}$ is any adapted random vector sequence satisfying
\begin{equation}
\sum_{i=1}^{n}\|x\|^{2}=O(n^{b}), \;\;\mathbb{E}\|x_{n}\|^{4}\{\log^{+}(\|x_{n}\|)\}^{2+\delta}=O(n^{2(b-1)}),
\end{equation}
where $\delta>0, b\geq 1$ are some costants.
Then as $n\rightarrow \infty$:
	\begin{equation}\label{19}
		\max_{1\leq t \leq n}\max_{1\leq i\leq n}\left\|\sum_{j=1}^{i}x_{j-t}w_{j}\right\|= O\left(n^{\frac{b}{2}}\{\log n\}^{\frac{1}{2}}\right)\;a.s., \forall \eta >0.
	\end{equation}
\end{lemma} 

\end{appendix}

\end{document}